\documentclass{spie}%
\usepackage{amsfonts}
\usepackage{amsmath}
\usepackage{amssymb}
\usepackage{graphicx}%
\providecommand{\U}[1]{\protect\rule{.1in}{.1in}}

\newtheorem{corollary}[theorem]{Corollary}
\newtheorem{proposition}[theorem]{Proposition}
\begin{document}

\title{A 3-Stranded Quantum Algorithm for the Jones Polynomial}
\author{Louis H. Kauffman\supit{a} and Samuel J. Lomonaco, Jr\supit{b}
\skiplinehalf\supit{a}Department of Mathematics, Statistics, and Computer
Science, 851 South Morgan Street, University of Illinois at Chicago,
Chicago,IL \ 60607-7045, USA\\\supit{b}Department of Computer Science and Electrical Engineering, 1000
Hilltop Circle, University of Maryland Baltimore County (UMBC), Baltimore, MD
\ 21250, USA }
\authorinfo{L.H.K.: Email: kauffman@uic.edu, S.J.L.: Email: lomonaco@umbc.edu}
\authorinfo{L.H.K.: E-mail: kauffman@uic.edu,}
\authorinfo{L.H.K.: E-mail: kauffman@uic.edu, S.J.L.: E-mail: lomonaco@umbc.edu}
\maketitle

\begin{abstract}
Let $K$ be a 3-stranded knot (or link), and let $L$ denote the number of
crossings in $K$. \ Let $\epsilon_{1}$ and $\epsilon_{2}$ be two positive real
numbers such that $\epsilon_{2}\leq1$. \ 

In this paper, we create two algorithms for computing the value of the Jones
polynomial $V_{K}\left(  t\right)  $ at all points $t=\exp\left(
i\varphi\right)  $ of the unit circle in the complex plane such that
$\left\vert \varphi\right\vert \leq2\pi/3$. \ 

The first algorithm, called the \textbf{classical 3-stranded braid (3-SB)
algorithm}, is a classical deterministic algorithm that has time complexity
$O\left(  L\right)  $. \ The second, called the \textbf{quantum 3-SB
algorithm}, is a quantum algorithm that computes an estimate of $V_{K}\left(
\exp\left(  i\varphi\right)  \right)  $ within a precision of $\epsilon_{1}$
with a probability of success bounded below by $1-\epsilon_{2}$. \ The
\textbf{execution time complexity} of this algorithm is $O\left(  nL\right)
$, where $n$ is the ceiling function of $\ \left(  \ln\left(  4/\epsilon
_{2}\right)  \right)  /2\epsilon_{1}^{2}$. \ The \textbf{compilation time
complexity}, i.e., an asymptotic measure of the amount of time to assemble the
hardware that executes the algorithm, is $O\left(  L\right)  $.

\end{abstract}

\section{Introduction}

\bigskip

Let $K$ be a 3-stranded knot (or link), i.e., a knot formed by the closure
$\overline{b}$ of a 3-stranded braid $b$, i.e., a braid $b\in B_{3}$. \ Let
$L$ be the length of the braid word $b$, i.e., the number of crossings in the
knot (or link) $K$. \ Let $\epsilon_{1}$ and $\epsilon_{2}$ be two positive
real numbers such that $\epsilon_{2}\leq1$. \ 

\bigskip

In this paper, we create two algorithms for computing the value of the Jones
polynomial $V_{K}\left(  t\right)  $ at all points $t=e^{i\varphi}$ of the
unit circle in the complex plane such that $\left\vert \varphi\right\vert
\leq\frac{2\pi}{3}$. \ 

\bigskip

The first algorithm, called the \textbf{classical 3-stranded braid (3-SB)
algorithm}, is a classical deterministic algorithm that has time complexity
$O\left(  L\right)  $. \ The second, called the \textbf{quantum 3-SB
algorithm}, is a quantum algorithm that computes an estimate of $V_{K}\left(
e^{i\varphi}\right)  $ within a precision of $\epsilon_{1}$ with a probability
of success bounded below by $1-\epsilon_{2}$. \ The \textbf{execution time
complexity} of this algorithm is $O\left(  nL\right)  $, where $n$ is the
ceiling function of $\ \frac{\ln\left(  4/\epsilon_{2}\right)  }{2\epsilon
_{1}^{2}}$. \ The \textbf{compilation time complexity}, i.e., an asymptotic
measure of the amount of time to assemble the hardware that executes the
algorithm, is $O\left(  L\right)  $. \ 

\bigskip

\section{The braid group}

\bigskip

The the $n$\textbf{-stranded braid group} $B_{n}$ is the group generated by
the symbols
\[
b_{1}\text{, }b_{2}\text{, }\ldots\text{ , }b_{n-1}%
\]
subject to the following complete set of defining relations%
\[
\left\{
\begin{array}
[c]{ll}%
b_{i}b_{j}=b_{j}b_{i} & \text{for }\left\vert i-j\right\vert >1\\
& \\
b_{i}b_{i+1}b_{i}=b_{i+1}b_{i}b_{i+1} & \text{for }1\leq i<n
\end{array}
\right.
\]

\bigskip

This group can be described more informally in terms of diagrammatics as
follows: \ We think of each braid as a hatbox with $n$ black dots on top and
another $n$ on the bottom, and with each top black dot connected by a red
string (i.e., a strand) to a bottom black dot. \ The strands are neither
permitted to intersect nor to touch. \ Two such hatboxes (i.e., braids) are
said to be equal if it is possible to continuously transform the strands of
one braid into those of the other, without leaving the hatbox, without cutting
and reconnecting the strands, and without permitting one strand to pass
through or touch another. \ The product of two braids $b$ and $b^{\prime}$ is
defined by simply stacking the hatbox $b$ on top of the hatbox $b^{\prime}$,
thereby producing a new braid \ $b\cdot b^{\prime}$. \ Please refer to Figure
1. \ The generators $b_{i}$ are illustrated in Figure 2. \ Moreover, the
defining relations for the braid group $B_{n}$ are shown in Figures 3.\ \ The
reader should take care to note that the hatbox is frequently not drawn, but
is nonetheless understood to be there. \ 

\bigskip%

\begin{center}
\includegraphics[
height=1.6336in,
width=2.8824in
]%
{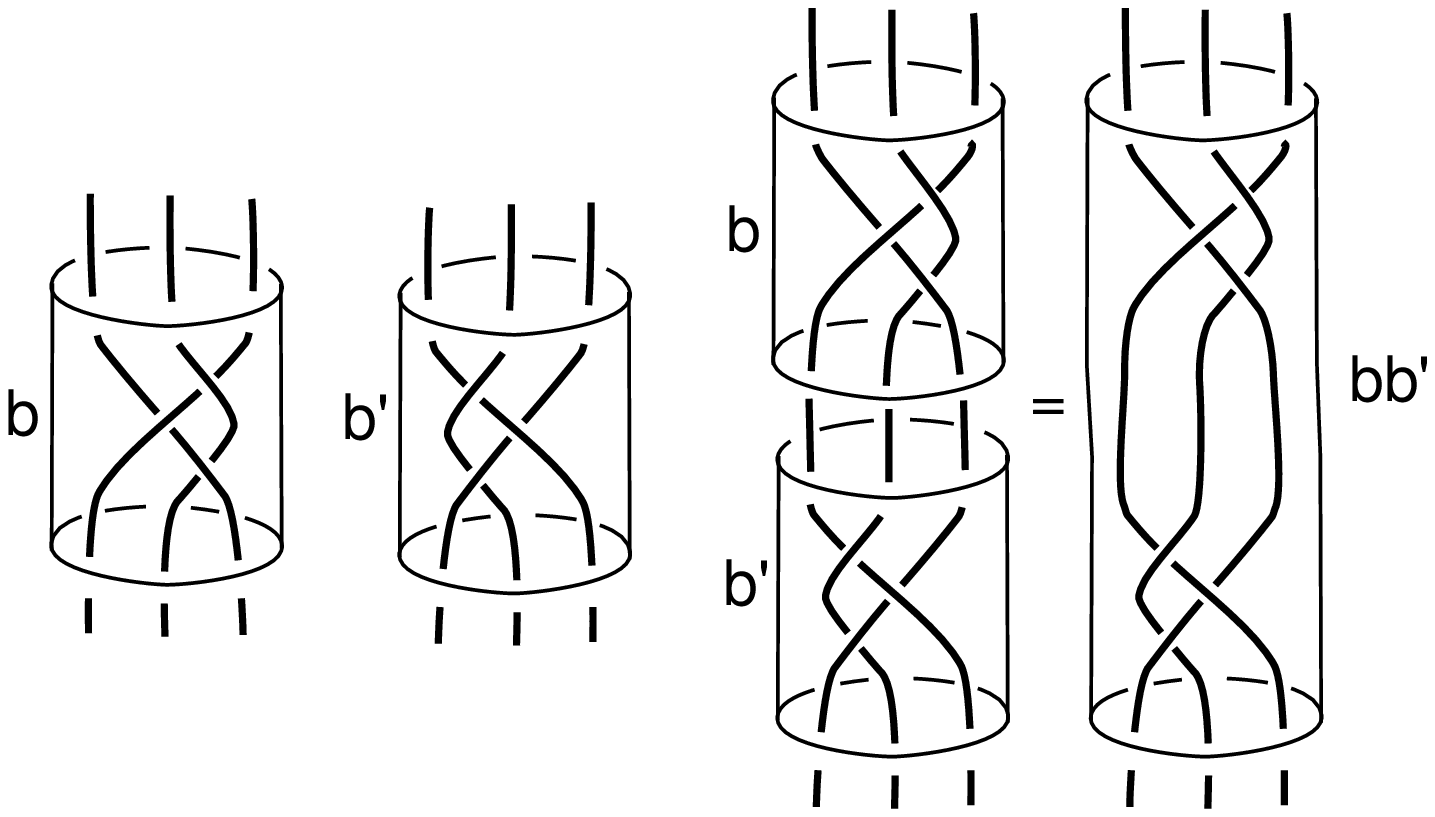}%
\\
\textbf{Figure 1. The product of two braids}%
\end{center}

\bigskip%

\begin{center}
\includegraphics[
height=0.8977in,
width=1.9501in
]%
{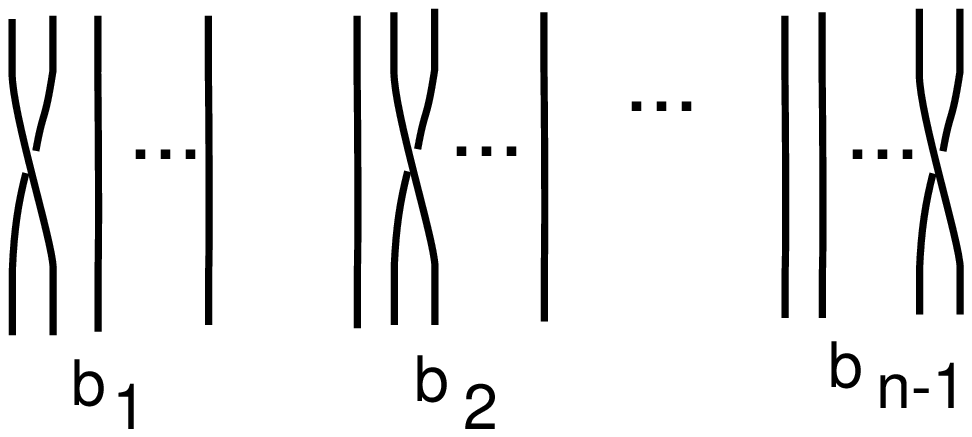}%
\\
\textbf{Figure 2. \ The generators of the }$n$-stranded \textbf{braid group
}$B_{n}$.
\end{center}

\bigskip%

\begin{center}
\includegraphics[
height=1.6708in,
width=2.418in
]%
{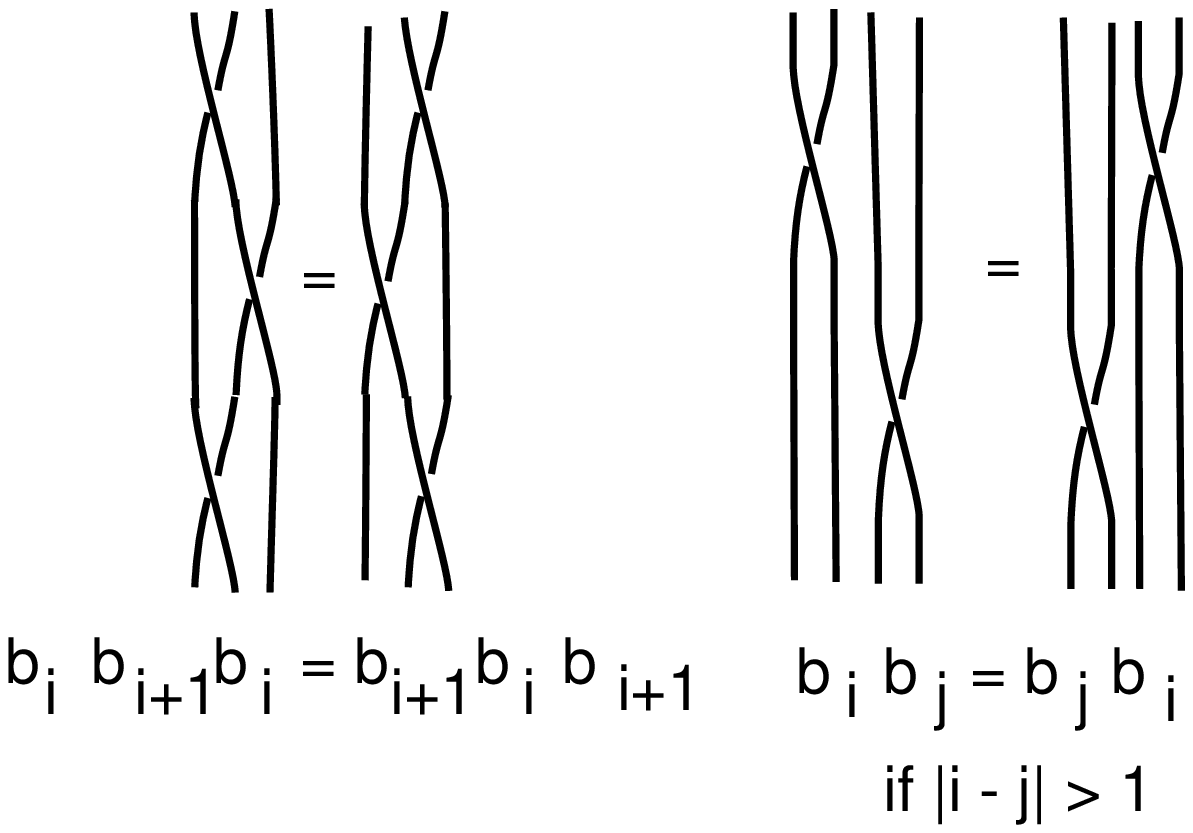}%
\\
\textbf{Figure 3. \ A complete set of defining relations for the braid group
}$B_{n}$.
\end{center}

\bigskip

Every braid $b$ in the braid group $B_{n}$ can be written as a product of
braid generators $b_{1}$, $b_{2}$, $\ldots$ , $b_{n-1}$ and their inverses
$b_{1}^{-1}$, $b_{2}^{-1}$, $\ldots$ , $b_{n-1}^{-1}$, i.e., every braid $b$
can be written in the form%
\[
b=%
{\displaystyle\prod\limits_{i=1}^{L}}
b_{j\left(  i\right)  }^{\epsilon\left(  i\right)  }=b_{j\left(  1\right)
}^{\epsilon\left(  1\right)  }b_{j\left(  2\right)  }^{\epsilon\left(
2\right)  }\cdots b_{j\left(  L\right)  }^{\epsilon\left(  L\right)  }\text{
,}%
\]
where $\epsilon\left(  i\right)  =\pm1$. \ We call such a product a
\textbf{braid word}. \ 

\bigskip

\noindent\textbf{Remark.} \textit{We will later see that each such braid word
can be thought of as a computer program which is to be compiled into an
executable program. This resulting compiled program will in turn be executed
to produce an approximation of the value of the Jones polynomial }%
$J_{K}\left(  t\right)  $\textit{ at a chosen point }$e^{i\varphi}$\textit{ on
the unit circle.}

\bigskip

We define

\bigskip

\begin{definition}
The \textbf{writhe} of a braid $b$, written $Writhe(b)$, is defined as the sum
of the exponents of a braid word representing the braid. \ In other words,
\[
Writhe\left(
{\displaystyle\prod\limits_{i=1}^{L}}
b_{j\left(  i\right)  }^{\epsilon\left(  i\right)  }\right)  =%
{\displaystyle\sum\limits_{i=1}^{L}}
\epsilon\left(  i\right)
\]

\end{definition}

\bigskip

For readers interested in learning more about the braid group, we refer the
reader to Emil Artin's original defining papers\cite{Artin1},\ \cite{Artin2}
\ \cite{Artin3} as well as to the many books on braids and knot theory, such
as for example\cite{Birman1}. \cite{Crowell1} \ \cite{Kauffman1}
\ \cite{Murasugi1}

\bigskip

\section{How knots and braids are related}

\bigskip

As one might suspect, knots and braids are very closely related to one
another. \ 

\bigskip

Every braid $b$ can be be transformed into a knot $K$ by forming the
\textbf{closed braid} $\overline{b}$ as shown in Figure 4. \ \
\begin{center}
\includegraphics[
height=1.6751in,
width=1.4641in
]%
{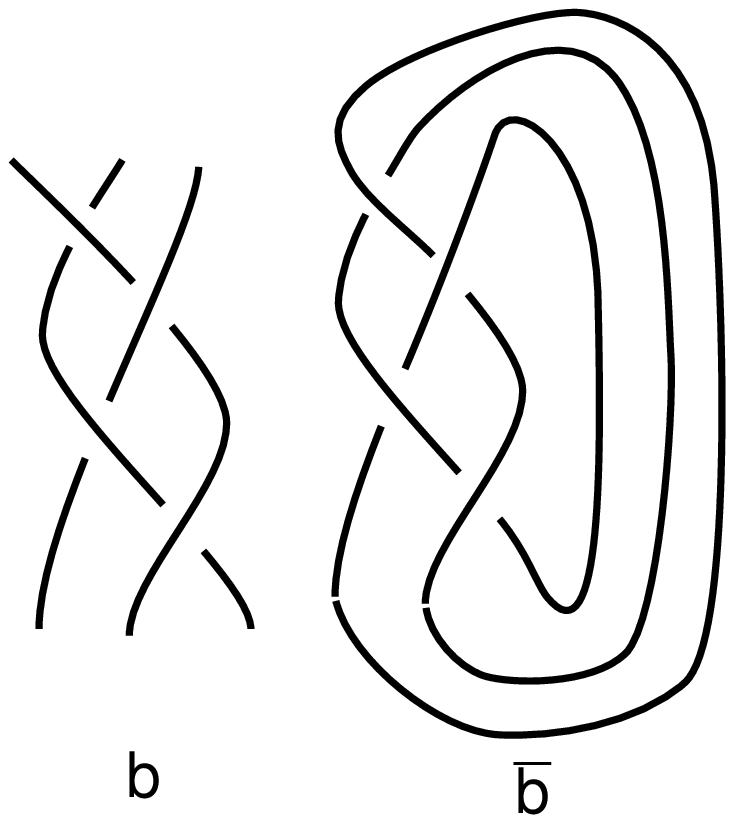}%
\\
Figure 4. A braid $b$ together with its closure $\overline{b}$.
\end{center}
This process can also be reversed. \ For Alexander developed a polytime
algorithm for transforming an arbitrary knot $K$ into a braid $b$ having $K$
as its closure.

\bigskip

\begin{theorem}
[Alexander]Every knot (or link) is the closure of a braid. \ Such a braid can
be found by a polynomial time algorithm
\end{theorem}

\bigskip

\noindent\textbf{Remark.} \textit{Every gardener who neatly puts away his
garden hose should no doubt be familiar with this algorithm.}

\bigskip

We should mention that it is possible that the closures of two different
braids will produce the same knot. \ But this non-uniqueness is well understood.

\bigskip

\begin{theorem}
[Markov]Two braids under braid closure produce the same knot (or link) if and
only if one can be transformed into the other by applying a finite sequence of
Markov moves
\end{theorem}

\bigskip

We will not describe the Markov moves in this paper. \ For the reader
interested in learning more about these moves, we suggest any one of the many
books on knot theory.\cite{Kauffman1} \ \cite{Murasugi1}

\bigskip

\section{The Temperley-Lieb algebra}

\bigskip

Let $d$ and $A$ be indeterminate complex numbers such that $d=-A^{2}-A^{-2}$,
and let
\[
\mathbb{Z}\left[  A,A^{-1}\right]
\]
be the ring of Laurent polynomials with integer coefficients in the
indeterminate $A$. \ Then the \textbf{Temperley-Lieb algebra} $TL_{n}\left(
d\right)  $ is the algebra with identity $1$ over the Laurent ring
$\mathbb{Z}\left[  A,A^{-1}\right]  $ generated by%
\[
1,U_{1},U_{2},\ldots,U_{n-1}%
\]
subject to the following complete set of defining relations%
\[
\left\{
\begin{array}
[c]{lc}%
U_{i}U_{j}=U_{j}U_{i} & \text{for }\left\vert i-j\right\vert >1\\
& \\
U_{i}U_{i\pm1}U_{i}=U_{i} & \\
& \\
U_{i}^{2}=dU_{i} &
\end{array}
\right.
\]

\bigskip

This algebra can be described more informally in much the same fashion as we
did for the braid group: \ We think of the generators $1,U_{1},U_{2}%
,\ldots,U_{n-1}$ as rectangles with $n$ top and $n$ bottom black dots, and
with $n$ disjoint red strings (i.e., strands) connecting distinct pairs of
black points. \ The red strings are neither permitted to intersect nor to
touch one another. However, they are now allowed to connect two top black dots
or two bottom black dots, as well as connect a top black dot with a bottom
black dot.\ \ The generators $1,U_{1},U_{2},\ldots,U_{n-1}$ of the
Temperley-Lieb algebra $T_{n}(d)$ are shown in Figure 5. \ The reader should
take care to note that the rectangle is frequently not drawn, but is
nonetheless understood to be there.%
\begin{center}
\includegraphics[
height=0.7801in,
width=3.141in
]%
{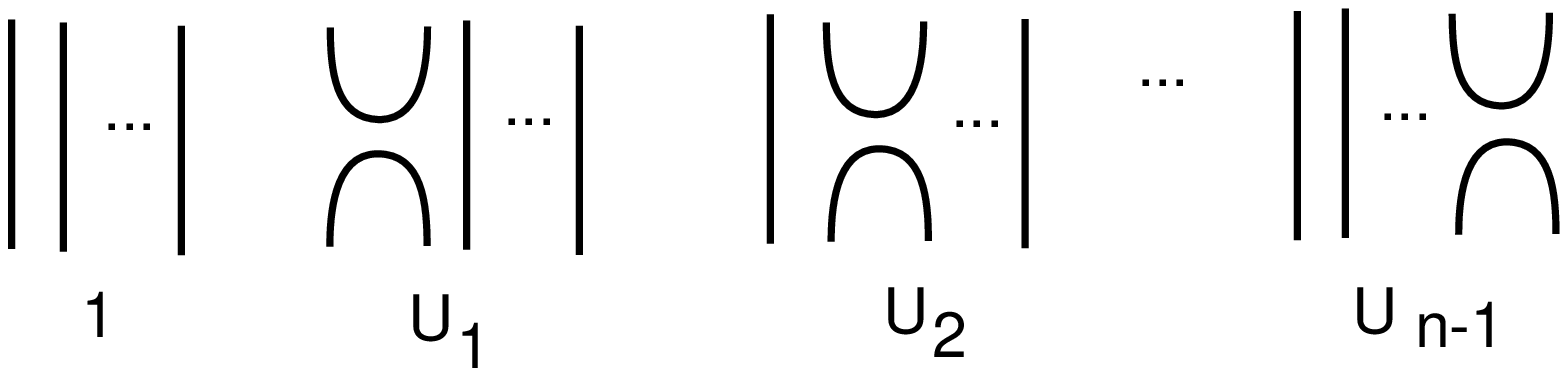}%
\\
\textbf{Figure 5. \ The generators of the Temperley-Lieb algebra }$TL_{n}%
(d)$\textbf{.}%
\end{center}

\bigskip

As we did with braids, the product `$\cdot$' of two such red stringed
rectangles is defined simply by stacking one rectangle on top of another.
\ However, unlike the braid group, there is one additional ingredient in the
definition of the product. Each disjoint circle resulting from this process is
removed from the rectangle, and replaced by multiplying the rectangle by the
indeterminate $d$. In this way, we can construct all the\ red stringed boxes
corresponding to all possible finite products of the generators $1,U_{1}%
,U_{2},\ldots,U_{n-1}$. \ As before, two such red stringed rectangles are said
to be equal if it is possible to continuously transform the strands of one
rectangle into those of the other, without leaving the rectangle, without
cutting and reconnecting the strands, and without letting one strand pass
through another. \ Please refer to Figure 6. \
\begin{center}
\includegraphics[
height=2.143in,
width=2.8349in
]%
{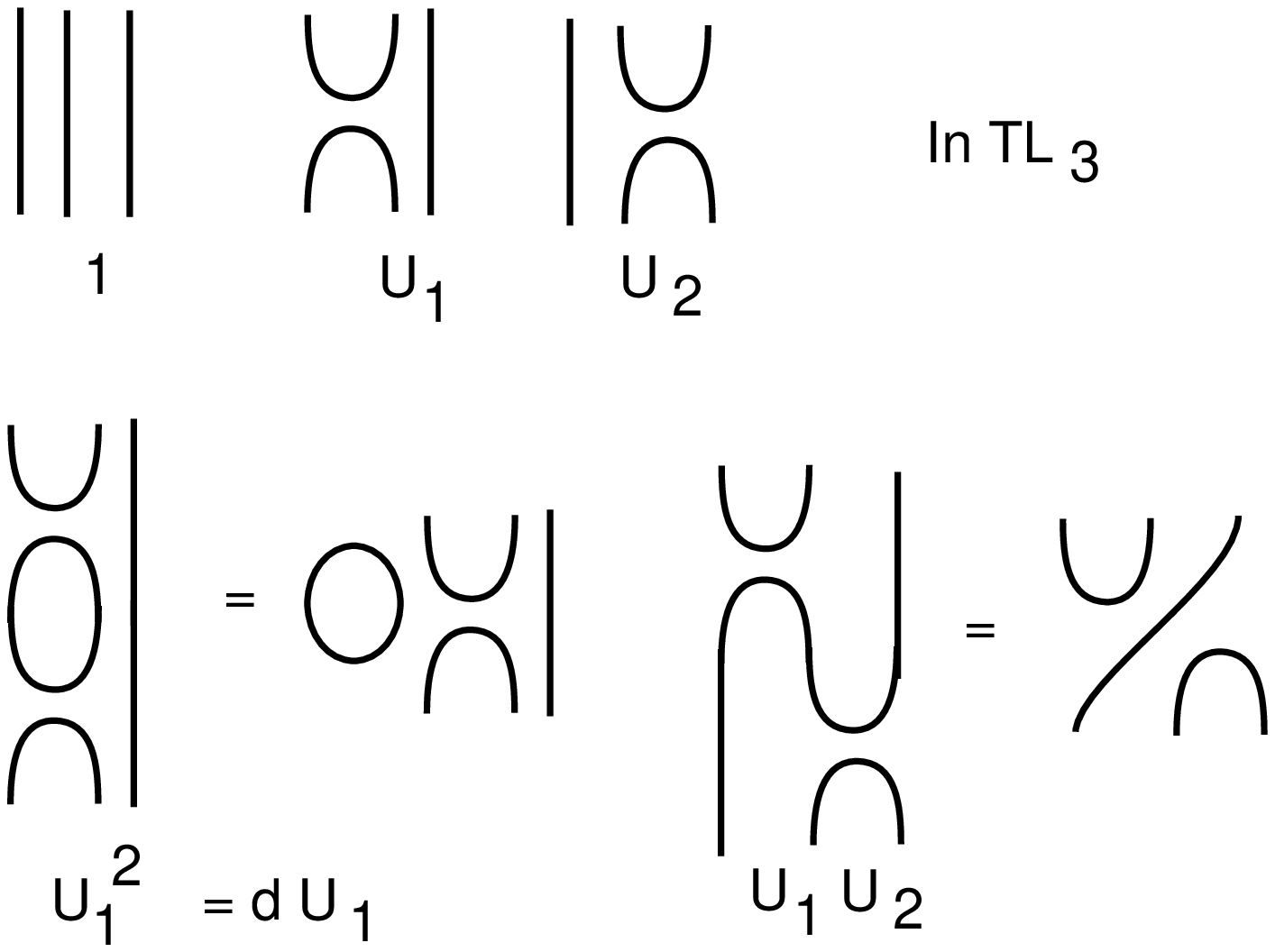}%
\\
\textbf{Figure 6. \ Two examples of the product of Temperley-Lieb generators.}%
\end{center}

\bigskip

Since $TL_{n}\left(  d\right)  $ is an algebra, we also need to define what is
meant by the sum `$+$' (linear combination) of two or more rectangles. \ This
is done simply by formally writing down linear combinations of rectangles over
the Laurent ring $\mathbb{Z}\left[  A,A^{-1}\right]  $, and then assuming that
addition `$+$' distributes with respect to the product `$\cdot$', and that the
scalar elements, i.e., the elements of the Laurent ring $\mathbb{Z}\left[
A,A^{-1}\right]  $, commute with all the rectangles and all the formal linear
combinations of these rectangles. An example of one such linear combination
is,%
\[
\left(  2A^{2}-3A^{-4}\right)  1+\left(  -5+7A^{2}\right)  U_{1}+\left(
1+A^{-6}-A^{-10}\right)  U_{1}U_{2}\text{ ,}%
\]

\bigskip

We should also mention that there exists a trace
\[
Tr_{M}:TL_{n}\left(  d\right)  \longrightarrow\mathbb{Z}\left[  A,A^{-1}%
\right]  \text{ ,}%
\]
called the \textbf{Markov trace}, from the Temperley-Lieb algebra
$TL_{n}\left(  d\right)  $ into the Laurent ring $\mathbb{Z}\left[
A,A^{-1}\right]  $. \ This trace is defined by sending each rectangle to
$d^{k-1}$, where $k$ denotes the number of disjoint circles that occur when
the closure of the rectangle is taken as indicated in Fig. 7.\bigskip

For readers interested in learning more about the Temperley-Lieb algebra
$TL_{n}\left(  d\right)  $, we refer them to the many books on knot theory,
such as for example\cite{Kauffman1}. \cite{Kauffman2}%
\begin{center}
\includegraphics[
height=3.1877in,
width=2.7812in
]%
{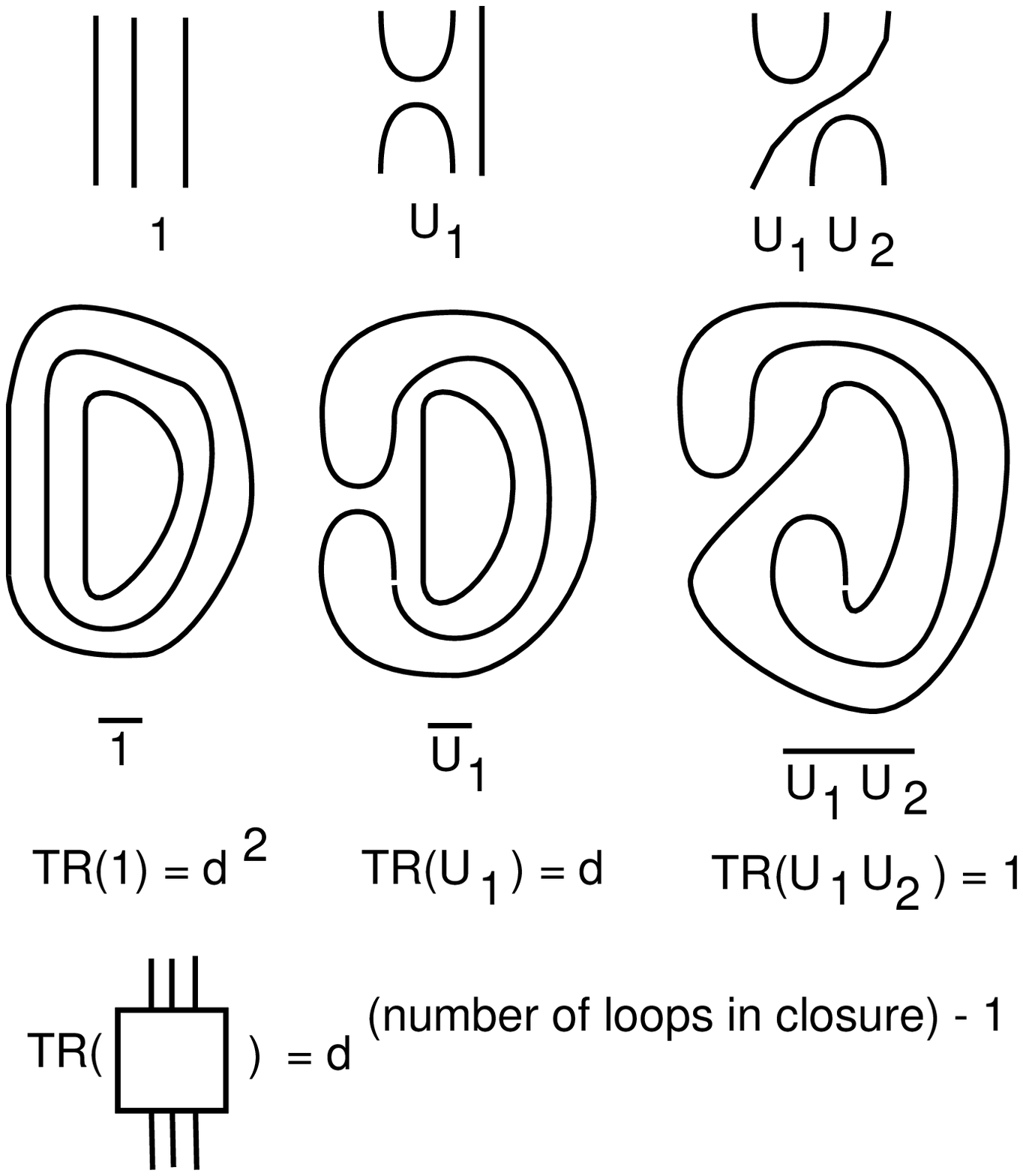}%
\\
\textbf{Figure 7. \ The Markov trace }$Tr_{M}:TL_{n}(d)\longrightarrow
Z\left[  A,A^{-1}\right]  $\textbf{.}%
\end{center}

\bigskip

\section{The Jones Representation}

\bigskip

Vaughn Jones, using purely algebraic methods, constructed his \textbf{Jones
representation}
\[
J:B_{n}\longrightarrow TL_{n}\left(  d\right)
\]
of the braid group $B_{n}$ into the Temperley-Lieb algebra $TL_{n}\left(
n\right)  $ by mapping each braid generator $b_{i}$ and its inverse
$b_{i}^{-1}$ into $TL_{n}\left(  d\right)  $ as follows\footnote{Actually to
be perfectly correct, Jones wrote his original representation in a variable
$t$ which is related to our variable $A$ by the equation $t=A^{-4}$.}%
\[
\left\{
\begin{array}
[c]{ccc}%
b_{i} & \longmapsto & A1+A^{-1}U_{i}\\
&  & \\
b_{i}^{-1} & \longmapsto & A^{-1}1+AU_{i}%
\end{array}
\right.
\]
\ He then used his representation $J$ and the Markov trace $Tr_{M}$ to
construct the \textbf{Jones polynomial} $V\left(  t\right)  $ of a knot $K$
(given by the closure $\overline{b}$ of a braid $b$) as
\[
V\left(  t\right)  =\left(  -A^{3}\right)  ^{Writhe(b)}Tr_{M}\left(  J\left(
b\right)  \right)
\]
where $t=A^{-4}$.

\bigskip

Later, Kauffman created the now well known diagrammatic approach to the
Temperley-Lieb algebra $TL_{n}\left(  d\right)  $ and showed that his
\textbf{bracket polynomial} $\left\langle \overline{b}\right\rangle $ was
intimately connected to the Jones polynomial via the formula
\[
\left\langle \overline{b}\right\rangle =Tr_{M}\left(  J\left(  \overline
{b}\right)  \right)
\]

\bigskip

For readers interested in learning more about these topic, we refer them to
the many books on knot theory, such as for example\cite{Jones1}. \cite{Jones2}
\ \cite{Kauffman1} \ \cite{Kauffman2} \ \cite{Murasugi1}

\bigskip

\section{The Temperley-Lieb algebra $TL_{3}\left(  d\right)  $}

\bigskip

We now describe a method for creating degree two representations of the
Temperley-Lieb algebra $TL_{3}\left(  d\right)  $. \ These representation will
in turn be used to create a unitary representation of the braid group $B_{3}$,
and ultimately be used to construct a quantum algorithm for computing
approximations of the values of the Jones polynomial on a large portion of the
unit circle in the complex plane.

\bigskip

From a previous section of this paper, we know that the 3 stranded
Temperley-Lieb algebra $TL_{3}\left(  d\right)  $ is generated by
\[
1,U_{1},U_{2}%
\]
with the complete set of defining relations given by
\[
\left\{
\begin{array}
[c]{lcl}%
U_{1}^{2}=dU_{1} & \text{ \ \ and \ \ } & U_{2}^{2}=dU_{2}\\
&  & \\
U_{1}U_{2}U_{1}=U_{1} & \text{ \ \ and \ \ } & U_{2}U_{1}U_{2}=U_{2}%
\end{array}
\right.
\]
Moreover, the reader can verify the following proposition.

\bigskip

\begin{proposition}
The elements
\[
1,U_{1},U_{2},U_{1}U_{2},U_{2}U_{1}%
\]
form a basis of $TL_{3}\left(  d\right)  $ as a module over the ring
$\mathbb{Z}\left[  A,A^{-1}\right]  $. \ In other words, every element
$\omega$ of $TL_{3}\left(  d\right)  $ can be written as a linear combination
of the form
\[
\omega=\omega_{0}1+\omega_{1}U_{1}+\omega_{2}U_{2}+\omega_{12}U_{1}%
U_{2}+\omega_{21}U_{2}U_{1}=\omega_{0}1+\omega_{+}\text{ ,}%
\]
where
\[
\omega_{0},\omega_{1},\omega_{2},\omega_{12},\omega_{21}%
\]
are uniquely determined elements of the ring $\mathbb{Z}\left[  A,A^{-1}%
\right]  $.
\end{proposition}

\bigskip

\section{A degree 2 representation of the Temperley-Lieb algebra
$TL_{3}\left(  d\right)  $}

\bigskip

We construct a degree 2 representation of the Temperley-Lieb algebra
$TL_{3}\left(  d\right)  $ as follows:

\bigskip

Let $\left\vert e_{1}\right\rangle $ and $\left\vert e_{2}\right\rangle $ be
non-orthogonal unit length vectors from a two dimensional Hilbert space
$\mathcal{H}$. \ From Schwartz's inequality, we immediately know that%
\[
0<\left\vert \left\langle e_{1}|e_{2}\right\rangle \right\vert \leq1
\]

\bigskip

Let $\delta=\pm\left\vert \left\langle e_{1}|e_{2}\right\rangle \right\vert
^{-1}$. \ It immediately follows that
\[
1\leq\left\vert \delta\right\vert <\infty
\]
Moreover, let $\alpha$ denote a complex number such that $\delta=-\alpha
^{2}-\alpha^{-2}$.

\bigskip

We temporarily digress to state a technical lemma that will be needed later in
this paper. \ We leave the proof as an exercise for the reader.

\bigskip

\begin{lemma}
Let $\delta$ be a real number of magnitude $\left\vert \delta\right\vert
\geq1$, and let $\alpha$ be a complex number such that $\delta=-\alpha
^{2}-\alpha^{-2}$. \ Then each of the following is a necessary and sufficient
condition for $\alpha$ to lie on the unit circle:
\[%
\begin{tabular}
[c]{l}%
$\bullet$ \ $\delta$ is a real number such that $1\leq\left\vert
\delta\right\vert \leq2$.\\
$\bullet$ \ T\textit{here exist a} $\theta\in\lbrack0,2\pi]$ \textit{such
that} $\delta=-2\cos\left(  2\theta\right)  $.
\end{tabular}
\ \
\]
Thus,
\[
\left\{  \alpha\in\mathbb{C}:\exists\ \delta\text{ such that }1\leq\left\vert
\delta\right\vert \leq2\text{ and }\delta=-\alpha^{2}-\alpha^{-2}\right\}
\]
is equal to the following set of points on the unit circle%
\[
\left\{  e^{i\theta}:\theta\in\left[  0,\frac{\pi}{6}\right]  \sqcup\left[
\frac{\pi}{3},\frac{2\pi}{3}\right]  \sqcup\left[  \frac{5\pi}{6},\frac{7\pi
}{6}\right]  \sqcup\left[  \frac{4\pi}{3},\frac{5\pi}{3}\right]  \sqcup\left[
\frac{11\pi}{6},2\pi\right]  \right\}
\]
Also, as $\delta$ ranges over all values such that $1\leq\left\vert
\delta\right\vert \leq2$, $\alpha^{-4}$ ranges over two thirds of the unit
circle, i.e.,
\[
\left\{
\begin{array}
[c]{c}%
\\
\end{array}
\alpha^{-4}:\exists\ \delta\text{ such that }1\leq\left\vert \delta\right\vert
\leq2\text{ and }\delta=-\alpha^{2}-\alpha^{-2}%
\begin{array}
[c]{c}%
\\
\end{array}
\right\}  =\left\{  \ e^{i\varphi}:\left\vert \varphi\right\vert \leq
\frac{2\pi}{2}\ \right\}
\]

\end{lemma}

\bigskip%

\begin{center}
\includegraphics[
height=2.2131in,
width=3.7766in
]%
{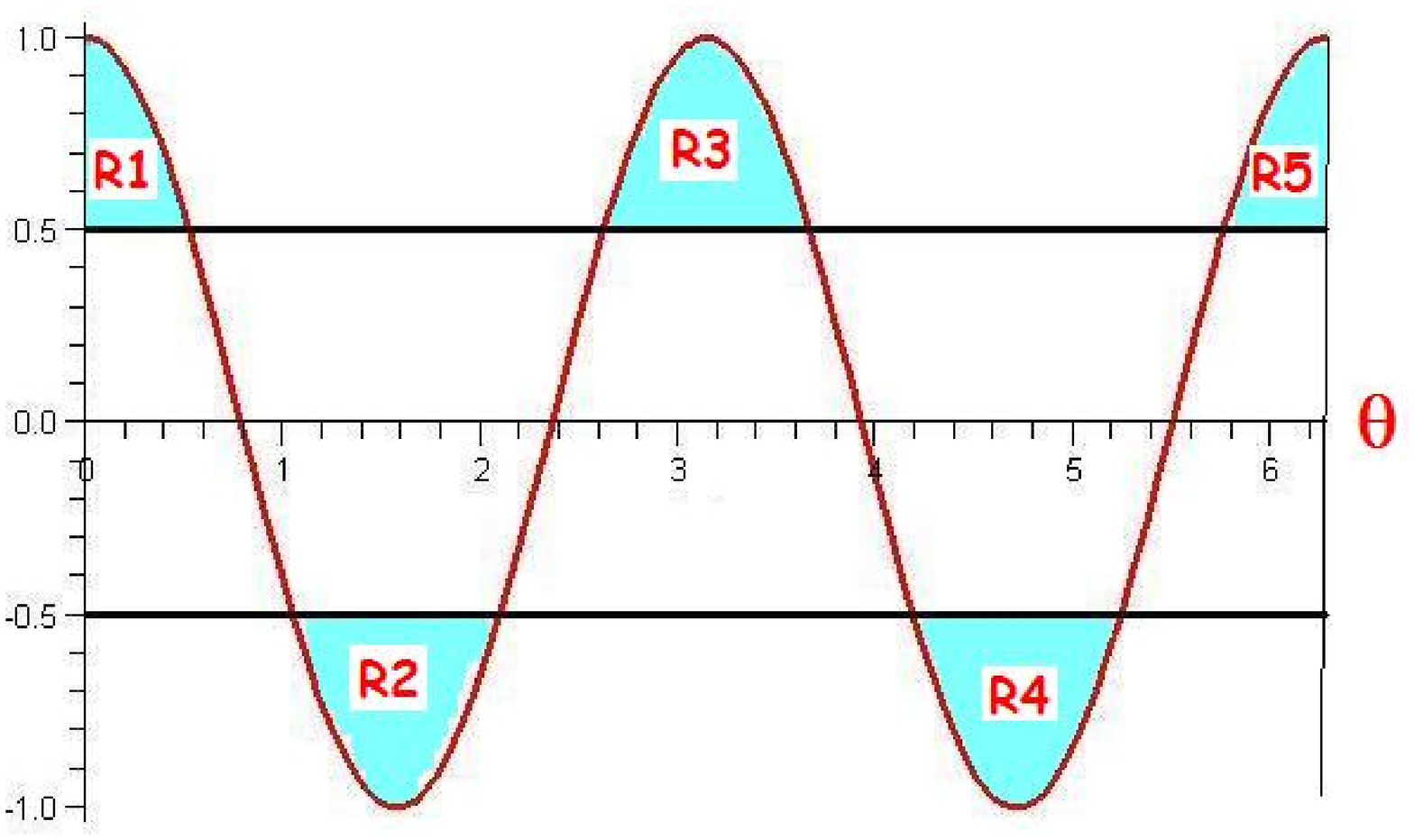}%
\\
Figure 8. \ A plot of $\cos\left(  2\theta\right)  $ for $0\leq\theta\leq2\pi
$.
\end{center}

\bigskip

We continue with the construction of our representation by using the unit
length vectors $\left\vert e_{1}\right\rangle $ and $\left\vert e_{2}%
\right\rangle $ to create projection operators
\[
E_{1}=\left\vert e_{1}\right\rangle \left\langle e_{1}\right\vert \text{
\ \ and \ \ }E_{2}=\left\vert e_{2}\right\rangle \left\langle e_{2}%
\right\vert
\]
These linear operators $E_{1}$ and $E_{2}$ are elements of the endomorphism
ring $End\left(  \mathcal{H}\right)  \cong Mat\left(  2,2;\mathbb{C}\right)  $
of the Hilbert space $\mathcal{H}$. \ Since they are projection operators,
they are Hermitian. \ By construction, they are of unit trace, i.e.,
\[
tr\left(  E_{1}\right)  =1=tr\left(  E_{2}\right)
\]
where $tr$ denotes the standard trace on $End\left(  \mathcal{H}\right)  \cong
Mat\left(  2,2;\mathbb{C}\right)  $. \ The reader can also readily verify
that
\[
tr(E_{1}E_{2})=\delta^{-2}=tr\left(  E_{2}E_{1}\right)
\]
and that $E_{1}$ and $E_{2}$ satisfy the relations

\[
\left\{
\begin{array}
[c]{lcl}%
E_{1}^{2}=E_{1} & \text{ \ \ and \ \ } & E_{2}^{2}=E_{2}\\
&  & \\
E_{1}E_{2}E_{1}=\delta^{-2}E_{1} & \text{ \ \ and \ \ } & E_{2}E_{1}%
E_{2}=\delta^{-2}E_{2}%
\end{array}
\right.
\]

\bigskip

It now follows that

\bigskip

\begin{theorem}
Let $\delta=\pm\left\vert \left\langle e_{1}|e_{2}\right\rangle \right\vert
^{-1}$ (hence, $\left\vert \delta\right\vert \geq1$), and let $\alpha$ be a
complex number such that $\delta=-\alpha^{2}-\alpha^{-2}$. \ Then the map%
\[%
\begin{array}
[c]{ccl}%
\Phi_{\alpha}:TL_{3}\left(  d\right)  & \longrightarrow & End\left(
\mathcal{H}\right)  \cong Mat\left(  2,2;\mathbb{C}\right) \\
U_{j} & \longmapsto & \quad\delta E_{j}\\
d & \longmapsto & \quad\ \delta\\
A & \longmapsto & \quad\ \alpha
\end{array}
\]
is a well defined degree 2 representation of the Temperley-Lieb algebra
$TL_{3}\left(  d\right)  $. \ Moreover, we have%
\[%
\begin{array}
[c]{c}%
tr\left(  \Phi_{\alpha}\left(  U_{1}\right)  \right)  =\delta=tr\left(
\Phi_{\alpha}\left(  U_{2}\right)  \right) \\
\text{and}\\
tr\left(  \Phi_{\alpha}\left(  U_{1}U_{2}\right)  \right)  =1=tr\left(
\Phi_{\alpha}\left(  U_{2}U_{1}\right)  \right)
\end{array}
\]

\end{theorem}

\bigskip

\begin{proposition}
Let $\delta=\pm\left\vert \left\langle e_{1}|e_{2}\right\rangle \right\vert
^{-1}$ (hence, $\left\vert \delta\right\vert \geq1$), and let $\alpha$ be a
complex number such that $\delta=-\alpha^{2}-\alpha^{-2}$. \ Moreover, let
$\operatorname{eval}_{\alpha}:$ $\mathbb{Z}\left[  A,A^{-1}\right]
\longrightarrow\mathbb{C}$ be the map defined by $A\longmapsto\alpha$. \ Then
the diagram
\[%
\begin{array}
[c]{ccc}%
TL_{3}\left(  d\right)  & \overset{\Phi_{a}}{\longrightarrow} & Mat\left(
2,2;\mathbb{C}\right) \\
Tr_{M}\downarrow &  & \downarrow tr\\
\mathbb{Z}\left[  A,A^{-1}\right]  & \overset{\operatorname{eval}_{\alpha}%
}{\longrightarrow} & \mathbb{C}%
\end{array}
\]
is \textbf{almost commutative} in the sense that, for each element $\omega\in
TL_{3}\left(  d\right)  $,%
\[
\operatorname{eval}_{\alpha}\circ Tr_{M}\left(  \omega\right)  =tr\circ
\Phi_{\alpha}\left(  \omega\right)  +\left(  \delta-2\right)  \omega
_{0}\text{, }%
\]
where $\omega_{0}$ denotes the coefficient of the generator $1$ in $\omega$.
\end{proposition}

\bigskip

\section{A degree 2 unitary representation of the three stranded braid group
$B_{3}$}

\bigskip

In this section, we compose the above constructed representation $\Phi
_{\alpha}$ with the Jones representation $J$ to create a representation of the
three stranded braid group $B_{3}$. \ We then determine when this
representation $\Phi_{\alpha}$ is unitary.

\bigskip

We begin by quickly recalling that the 3-stranded braid group $B_{3}$ is
generated by the standard braid generators
\[
b_{1},b_{2}%
\]
with, in this case, the single defining relation%
\[
b_{1}b_{2}b_{1}=b_{2}b_{1}b_{2}%
\]

\bigskip

We also recall that the Jones representation
\[
B_{3}\overset{J}{\longrightarrow}TL_{3}\left(  d\right)
\]
is defined by
\[
\left\{
\begin{array}
[c]{ccc}%
b_{j} & \longmapsto & A1+A^{-1}U_{j}\\
&  & \\
b_{j}^{-1} & \longmapsto & A^{-1}1+AU_{j}%
\end{array}
\right.
\]
where $A$ is a indeterminate satisfying $d=-A^{2}-A^{-2}$. \ Thus, if we let
$\delta$ and $\alpha$ be as defined in the previous section, we have that
\[%
\begin{array}
[c]{ccl}%
B_{3} & \overset{\Phi_{\alpha}\circ J}{\longrightarrow} & End\left(
\mathcal{H}\right)  \cong Mat\left(  2,2;\mathbb{C}\right)  \\
b_{j} & \longmapsto & \alpha I+\alpha^{-1}dE_{j}\\
b_{j} & \longmapsto & \alpha^{-1}I+\alpha dE_{j}%
\end{array}
\]
is a degree 2 representation of the braid group $B_{3}$. \ Moreover, we have

\bigskip

\begin{proposition}
Let $\delta=\pm\left\vert \left\langle e_{1}|e_{2}\right\rangle \right\vert
^{-1}$ (hence, $\left\vert \delta\right\vert \geq1$), and let $\alpha$ be a
complex number such that $\delta=-\alpha^{2}-\alpha^{-2}$. \ Then the degree 2
representation
\[%
\begin{array}
[c]{ccl}%
B_{3} & \overset{\Phi_{\alpha}\circ J}{\longrightarrow} & End\left(
\mathcal{H}\right)  \cong Mat\left(  2,2;\mathbb{C}\right)
\end{array}
\]
is a unitary representation of the braid group $B_{3}$ if and only if $\alpha$
lies on the unit circle in the complex plane.
\end{proposition}

\begin{proof}
Since $d$ is real and $E_{j}$ is Hermitian, we have
\[
\left(  \alpha I+\alpha^{-1}\delta E_{j}\right)  ^{\dag}=\overline{\alpha
}I+\overline{\alpha}^{-1}\delta E_{j}%
\]
So for unitarity, we must have
\[
\overline{\alpha}I+\overline{\alpha}^{-1}\delta E_{j}=\alpha^{-1}%
I+\alpha\delta E_{j}%
\]

It now follows from the linear independence of $I$, $E_{1}$, $E_{2}$ that
$\Phi\circ J$ is unitary if and only if
\[
\overline{\alpha}=\alpha^{-1}%
\]

\end{proof}

\bigskip

From lemma 7.1, we have the following

\bigskip

\begin{corollary}
Let $\delta=\pm\left\vert \left\langle e_{1}|e_{2}\right\rangle \right\vert
^{-1}$ (hence, $\left\vert \delta\right\vert \geq1$), and let $\alpha$ be a
complex number such that $\delta=-\alpha^{2}-\alpha^{-2}$. \ Then the
representation $\Phi_{\alpha}\circ J$ is unitary if and only if $\alpha
=e^{i\theta}$, where $\theta$ lies in the set%
\[
\left\{  \theta\in\left[  0,2\pi\right]  :\left\vert \cos\left(
2\theta\right)  \right\vert \geq\frac{1}{2}\right\}  =\left[  0,\frac{\pi}%
{6}\right]  \sqcup\left[  \frac{\pi}{3},\frac{2\pi}{3}\right]  \sqcup\left[
\frac{5\pi}{6},\frac{7\pi}{6}\right]  \sqcup\left[  \frac{4\pi}{3},\frac{5\pi
}{3}\right]  \sqcup\left[  \frac{11\pi}{6},2\pi\right]
\]

\end{corollary}

\bigskip

\section{Computing the Jones polynomial}

\bigskip

Recall that the Jones polynomial $V\left(  t\right)  $ of a knot (or link) $K$
given by the closure $\overline{b}$ of a braid word $b$ is defined as
\[
V\left(  t\right)  =\left(  -A^{3}\right)  ^{Writhe(b)}Tr_{M}\left(  J\left(
b\right)  \right)  \text{ ,}%
\]
where $t=A^{-4}$. \ 

\bigskip

Thus, from Proposition 7.3, we know that the value of the Jones polynomial at
a point $t=e^{i\varphi}$ on the unit circle is given by%
\[
V\left(  e^{i\varphi}\right)  =\left(  -e^{3i\theta}\right)  ^{Writhe(b)}%
\operatorname{eval}_{e^{i\theta}}\circ Tr_{M}\circ J\left(  b\right)  =\left(
-e^{3i\theta}\right)  ^{Writhe(b)}\left(  tr\circ\Phi_{e^{i\theta}}\circ
J\right)  \left(  b\right)  +\left(  \delta-2\right)  \left(  -e^{4i\theta
}\right)  ^{Writhe(b)}\text{ \ ,}%
\]
where $e^{i\theta}$ is a point on the unit circle such that $e^{i\varphi
}=\left(  e^{i\theta}\right)  ^{-4}=e^{-4i\theta}$. \ From lemma 1, we know
that $\Phi_{e^{i\theta}}$ is only defined when $\left\vert \cos\left(
2\theta\right)  \right\vert \geq\frac{1}{2}$. \ \ Moreover, since
$\varphi=-4\theta\ \operatorname{mod}\ 2\pi$, it also follows from lemma 1
that $\Phi_{e^{i\theta}}$ is only defined when $\left\vert \varphi\right\vert
\leq\frac{2\pi}{3}$.

\bigskip

\begin{theorem}
Let $\varphi$ be a real number such that $\left\vert \varphi\right\vert
\leq\frac{2\pi}{3}$,\ and let $\theta$ be a real number such that
$\varphi=-4\theta\ \operatorname{mod}\ 2\pi$. \ Let $K$ be a knot (or link)
given by the closure $\overline{b}$ of a 3-stranded braid $b\in B_{3}$. \ Then
the value of the Jones polynomial $V\left(  t\right)  $ for the knot (or link)
$K$ at $t=e^{i\varphi}$ is given by%
\[
V\left(  e^{i\varphi}\right)  =tr\left(  U(b)\right)  +\left(  \delta
-2\right)  e^{i\theta Writhe(b)}\text{ ,}%
\]
where $U$ is the unitary transformation
\[
U=U(b)=\left(  \Phi_{e^{i\theta}}\circ J\right)  \left(  b\right)  \text{ \ .}%
\]

\end{theorem}

\bigskip

Let us now assume that $\left\vert \varphi\right\vert \leq\frac{2\pi}{3}$ and
that $\varphi=-4\theta\ \operatorname{mod}\ 2\pi$. \ Hence, $U=U(b)$ is
unitary. \ Thus, if the knot (or link) $K$ is given by the closure
$\overline{b}$ of a braid $b$ defined by a word
\[
b=%
{\displaystyle\prod\limits_{k=1}^{L}}
b_{j\left(  k\right)  }^{\epsilon\left(  k\right)  }=b_{j(1)}^{\epsilon
(1)}b_{j(2)}^{\epsilon(2)}\cdots b_{j(L)}^{\epsilon(L)}\text{ \ ,}%
\]
where $b_{1},b_{2}$ are the generators of the braid group $B_{3}$, and where
$\epsilon\left(  k\right)  =\pm1$ for $k=1,2,\ldots,L$, then the unitary
transformation $U=U\left(  b\right)  $ can be rewritten as
\[
U=%
{\displaystyle\prod\limits_{k=1}^{L}}
\left(  U^{\left(  j(k)\right)  }\right)  ^{\epsilon\left(  k\right)  }\text{
,}%
\]
where $U^{(j)}$ denotes the unitary transformation (called an
\textbf{elementary gate}) given by%
\begin{align*}
U^{(j)}  &  =\left(  \Phi_{e^{i\theta}}\circ J\right)  \left(  b_{j}\right)
=\Phi_{e^{i\theta}}\left(  A1+A-1U_{j}\right)  =e^{i\theta}I-2e^{-i\theta}%
\cos\left(  2\theta\right)  \left\vert e_{j}\right\rangle \left\langle
e_{j}\right\vert \\
& \\
&  =e^{i\theta}I-2e^{-i\theta}\cos\left(  2\theta\right)  E_{j}%
\end{align*}

\bigskip\bigskip

In summary, we have:

\bigskip

\begin{corollary}
Let $t=e^{i\varphi}$ be an arbitrary point on the unit circle in the complex
plane. \ Let $b$ be a 3-stranded braid (i.e., a braid $b$ in $B_{3}$) given by
a braid word
\[
b=%
{\displaystyle\prod\limits_{k=1}^{L}}
b_{j\left(  k\right)  }^{\epsilon\left(  k\right)  }=b_{j(1)}^{\epsilon
(1)}b_{j(2)}^{\epsilon(2)}\cdots b_{j(L)}^{\epsilon(L)}\text{ , }%
\]
and let $K$ be the knot (or link) given by the closure $\overline{b}$ of the
braid $b$. \ Then the value of the Jones polynomial $V\left(  t\right)  $ of
$K$ at $t=e^{i\varphi}$ is given by
\[
V\left(  e^{i\varphi}\right)  =\left(  \left(  -e^{3i\theta}\right)
^{\sum_{k=1}^{L}\epsilon\left(  k\right)  }\right)  tr\left(
{\displaystyle\prod\limits_{k=1}^{L}}
\left(  U^{\left(  j\left(  k\right)  \right)  }\right)  ^{\epsilon\left(
k\right)  }\right)  +\left(  \delta-2\right)  \left(  \left(  -e^{4i\theta
}\right)  ^{\sum_{k=1}^{L}\epsilon\left(  k\right)  }\right)  \text{ ,}%
\]
where $U^{(j)}$ ($j=1,2$) is the linear transformation
\[
U^{(j)}=e^{i\theta}I-2e^{-i\theta}\cos\left(  2\theta\right)  E_{j}\text{ ,}%
\]
where $I$ denotes the $2\times2$ identity matrix, and where $E_{j}$ is the
$2\times2$ Hermitian matrix $\left\vert e_{j}\right\rangle \left\langle
e_{j}\right\vert $. \ We have also shown that the linear transformations
$U^{(1)}$, $U^{(2)}$, and $U=\prod_{k=1}^{L}\left(  U^{\left(  j\left(
k\right)  \right)  }\right)  ^{\epsilon\left(  k\right)  }$ are unitary if and
only if $\left\vert \varphi\right\vert \leq\frac{2\pi}{3}$. \ When $\left\vert
\varphi\right\vert \leq\frac{2\pi}{3}$, we will call $U^{(1)}$ and $U^{(2)}$
\textbf{elementary gates}.
\end{corollary}

\bigskip

\noindent\textbf{Remark.} \textit{Thus, the task of determining the value of
the Jones polynomial at any point }$t=e^{i\varphi}$\textit{such that
}$\left\vert \varphi\right\vert \leq\frac{2\pi}{3}$\textit{ reduces to the
task of devising a quantum algorithm that computes the trace of the unitary
transformation}%
\[
U=U(b)=\prod_{k=1}^{L}\left(  U^{\left(  j(k)\right)  }\right)  ^{\epsilon
(k)}\text{ \ .}%
\]

\bigskip

\begin{corollary}
Let $K$ be a \textbf{3-stranded knot} (or link),i.e., a knot (or link) given
by the closure $\overline{b}$ of a 3-stranded braid $b$, i.e., a braid $b\in
B_{3}$. \ \ Then the formula found in the previous corollary gives a
\textbf{deterministic classical algorithm} for computing the value of the
Jones polynomial of $K$ \textbf{at all points of the unit circle} in the
complex plane of the form $e^{i\varphi}$, where $\left\vert \varphi\right\vert
\leq\frac{2\pi}{3}$. \ Moreover, the time complexity of this algorithm is
$O\left(  L\right)  $, where $L$ is the length of the word $b$, i.e., where
$L$ is the number of crossings in the knot (or link) $K$. \ We will call this
algorithm the \textbf{classical 3-stranded braid (3-SB) algorithm.}
\end{corollary}

\bigskip

\section{Trace estimation via the Hadamard test.}

\bigskip

In the past section, we have shown how to create the classical 3-SB algorithm
that computes the values of the Jones polynomial of a 3-stranded knot $K$ on
two thirds of the unit circle. \ In this section, we will show how to
transform this classical algorithm into a corresponding quantum algorithm.

\bigskip

\bigskip

We will now assume that $\left\vert \varphi\right\vert \leq\frac{2\pi}{3}$ so
that the elementary gates $U^{(1)}$and $U^{(2)}$, and also the gate
$U=U\left(  b\right)  =\prod_{k=1}^{L}\left(  U^{\left(  j\left(  k\right)
\right)  }\right)  ^{\epsilon\left(  k\right)  }$ are unitary. \ We know from
the previous section that all we need to do to create a quantum 3-SB algorithm
is to devise a quantum procedure for estimating the trace $trace\left(
U\right)  $ of the unitary transformation $U$. \ 

\bigskip

To accomplish this, we will use a trace estimation procedure called the
Hadamard test. \ 

\bigskip

Let $\mathcal{H}$ be the two dimensional Hilbert space associated with the the
unitary transformations $U$, and let $\left\{  \left\vert k\right\rangle
:k=0,1\right\}  $ be a corresponding chosen orthonormal basis. \ Moreover, let
$\mathcal{K}$ denote the two dimensional Hilbert space associated with an
ancillary qubit with chosen orthonormal basis $\left\{  \left\vert
0\right\rangle ,\left\vert 1\right\rangle \right\}  $. \ Then the trace
estimation procedure, called the \textbf{Hadamard test}, is essentially
defined by the two wiring diagrams found in Figures 9 and 10. \ The wiring
diagrams found in Figures 9 and 10 are two basic quantum algorithmic
primitives for determining respectively the real part $\operatorname{Re}%
\left(  trace(U)\right)  $ and the imaginary part $\operatorname{Im}\left(
trace(U)\right)  $ of the trace $trace\left(  U\right)  $ of $U$. \ The top
qubit in each of these wiring diagrams denotes the ancillary qubit, and the
bottom qubit $\left\vert k\right\rangle $ denotes a basis element of the
Hilbert space $\mathcal{H}$ associated with $U$. \ (The top wire in each
wiring diagram denotes the ancilla qubit.) \ Each box labeled by an `$H$'
denotes the Hadamard gate%
\[
H=\frac{1}{\sqrt{2}}\left(
\begin{array}
[c]{rr}%
1 & 1\\
1 & -1
\end{array}
\right)  \text{ .}%
\]
The box labeled by an `$S$' denotes the phase gate
\[
S=\left(
\begin{array}
[c]{cc}%
1 & 0\\
0 & i
\end{array}
\right)  \text{. }%
\]
And finally, the controlled-$U$ gate is given by the standard
notation.\bigskip%

\begin{center}
\includegraphics[
height=0.8916in,
width=3.7031in
]%
{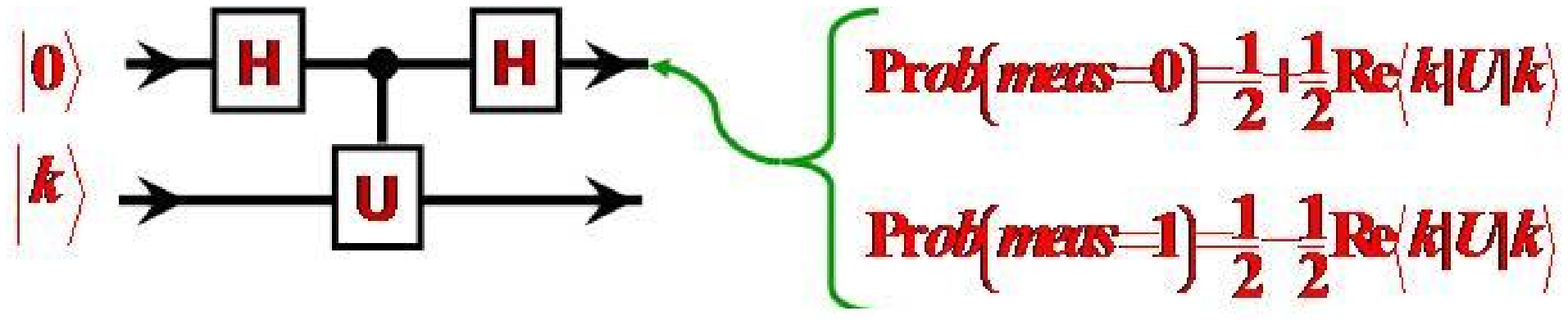}%
\\
Figure 9. \ A quantum system for computing the real part of the diagonal
element $U_{kk}$. $\operatorname{Re}\left(  U_{kk}\right)  =Prob\left(
meas=0\right)  -Prob\left(  meas=1\right)  $.
\end{center}
\bigskip%

\begin{center}
\includegraphics[
height=0.8536in,
width=3.7784in
]%
{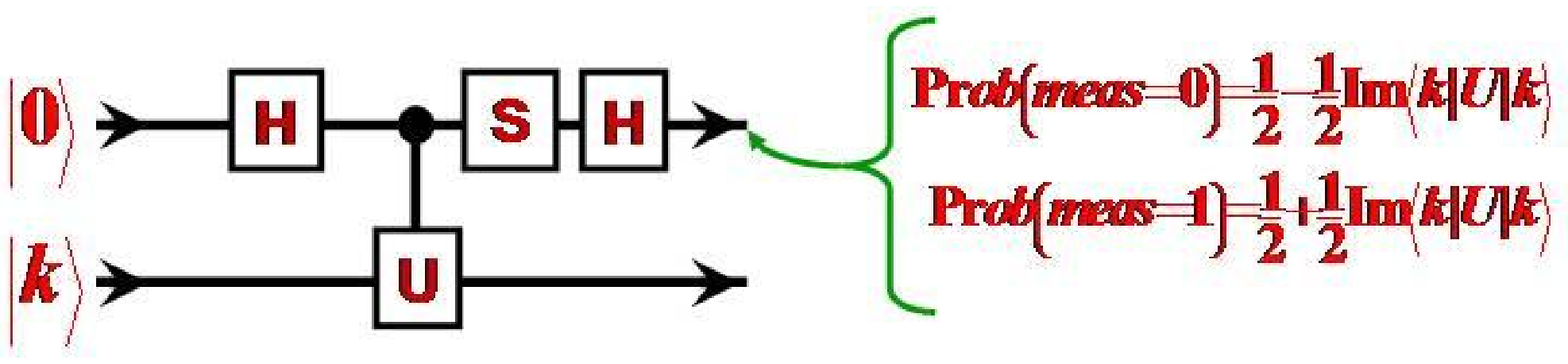}%
\\
Figure 10. \ A quantum system for computing the imaginary part of the diagonal
element $U_{kk}$. \ $\operatorname{Re}\left(  U_{kk}\right)  =Prob\left(
meas=1\right)  -Prob\left(  meas=0\right)  $.
\end{center}

\bigskip

\bigskip

The wiring diagram found in Figure 9 has been so designed as to compute the
real part $\operatorname{Re}\left(  U_{kk}\right)  $ of the $k$-th diagonal
entry $U_{kk}$ of $U$. \ For this wiring diagram has been so engineered that,
when the output ancilla qubit is measured, then the resulting measured $0$ or
$1$ occurs with probability given by
\[
\left\{
\begin{array}
[c]{ccc}%
Prob\left(  meas=0\right)  & = & \frac{1}{2}+\frac{1}{2}\operatorname{Re}%
\left\langle k|U|k\right\rangle \\
&  & \\
Prob\left(  meas=1\right)  & = & \frac{1}{2}-\frac{1}{2}\operatorname{Re}%
\left\langle k|U|k\right\rangle
\end{array}
\right.
\]
Thus, the difference of these two probabilities is the real part of the $k$-th
diagonal entry%
\[
Prob\left(  meas=0\right)  -Prob\left(  meas=1\right)  =\operatorname{Re}%
\left\langle k|U|k\right\rangle =\operatorname{Re}\left(  U_{kk}\right)
\text{ .}%
\]
If this procedure (i.e., preparation of the state $\left\vert 0\right\rangle
\left\vert k\right\rangle $, application of the unitary transformation
\ $\left(  H\otimes1\right)  \cdot Contr-U\cdot$\ $\left(  H\otimes1\right)
$, and measurement of the output ancilla qubit) is repeated $n$ times, then
the normalized number of $0$'s minus the number of $1$'s, i.e.,
\[
\frac{\#0\text{'s}-\#1\text{'s}}{n}\text{ ,}%
\]
becomes an ever better estimate of the real part $\operatorname{Re}\left(
U_{kk}\right)  $ of the $k$-th diagonal entry $U_{kk}$ as the number of trials
$n$ becomes larger and larger. \ We will make this statement even more precise later.

\bigskip

In like manner, the wiring diagram found in Figure 10 has been so designed to
compute the imaginary part $\operatorname{Im}\left(  U_{kk}\right)  $ of the
$k$-th diagonal entry $U_{kk}$ of $U$. \ This wiring diagram has been
engineered so that, if the output ancilla qubit is measured, then the
resulting measured $0$ and $1$ occur with probabilities given by
\[
\left\{
\begin{array}
[c]{ccc}%
Prob\left(  meas=0\right)  & = & \frac{1}{2}-\frac{1}{2}\operatorname{Re}%
\left\langle k|U|k\right\rangle \\
&  & \\
Prob\left(  meas=1\right)  & = & \frac{1}{2}+\frac{1}{2}\operatorname{Re}%
\left\langle k|U|k\right\rangle
\end{array}
\right.
\]
Thus, the difference of these two probabilities is the real part of the $k$-th
diagonal entry%
\[
Prob\left(  meas=1\right)  -Prob\left(  meas=0\right)  =\operatorname{Im}%
\left\langle k|U|k\right\rangle =\operatorname{Im}\left(  U_{kk}\right)
\]
Much like as before, if this procedure (i.e., preparation of the state
$\left\vert 0\right\rangle \left\vert k\right\rangle $, application of the
unitary transformation \ $\left(  H\otimes1\right)  \cdot Contr-U\cdot S\cdot
$\ $\left(  H\otimes1\right)  $, and measurement of the output ancilla qubit)
is repeated $n$ times, then the normalized number of $1$'s minus the number of
$0$'s, i.e.,
\[
\frac{\#1\text{'s}-\#0\text{'s}}{n}%
\]
becomes an ever better estimate of the imaginary part $\operatorname{Im}%
\left(  U_{kk}\right)  $ of the $k$-th diagonal entry $U_{kk}$ as the number
of trials $n$ increases.

\bigskip

We now focus entirely on the first wiring diagram, i.e., Figure 9. \ But all
that we will say can easily be rephrased for the second wiring diagram found
in Figure 10.

\bigskip

We continue by more formally reexpressing the wiring diagram of Figure 9 as
the quantum subroutine \textsc{QRe}$_{U}\left(  \ k\ \right)  $ given below:

\bigskip

\begin{center}
\fbox{\textsc{Quantum Subroutine} \textsc{QRe}$_{U}\left(  \ k\ \right)  $}
\end{center}

\bigskip

\begin{itemize}
\item[\fbox{\textbf{Step 0}}] Initialization%
\[
\left\vert \psi_{0}\right\rangle =\left\vert 0\right\rangle \left\vert
k\right\rangle
\]

\item[\fbox{\textbf{Step 1}}] Application of $H\otimes I$%
\[
\left\vert \psi_{1}\right\rangle =\left(  H\otimes I\right)  \left\vert
\psi_{0}\right\rangle =\frac{1}{\sqrt{2}}\left(  \left\vert 0\right\rangle
\left\vert k\right\rangle +\left\vert 1\right\rangle \left\vert k\right\rangle
\right)
\]

\item[\fbox{\textbf{Step 2}}] Application of \textsc{Contr-}$U$%
\[
\left\vert \psi_{2}\right\rangle =\left(  \text{\textsc{Contr-}}U\right)
\left\vert \psi_{1}\right\rangle =\frac{1}{\sqrt{2}}\left(  \left\vert
0\right\rangle \left\vert k\right\rangle +\left\vert 1\right\rangle
U\left\vert k\right\rangle \right)
\]

\item[\fbox{\textbf{Step 3}}] Application of $H\otimes I$%
\[
\left\vert \psi_{3}\right\rangle =\left(  H\otimes I\right)  \left\vert
\psi_{1}\right\rangle =\left\vert 0\right\rangle \left(  \frac{1+U}{2}\right)
\left\vert k\right\rangle +\left\vert 1\right\rangle \left(  \frac{1-U}%
{2}\right)  \left\vert k\right\rangle
\]

\item[\fbox{\textbf{Step 4}}] Measure the ancilla qubit%
\[%
\begin{tabular}
[c]{|c|c|c|}\hline
\textbf{Resulting Measured Bit }$\mathbf{b}$ & \textbf{Probability} &
\textbf{Resulting State} $\left\vert \mathbf{\psi}_{4}\right\rangle
$\\\hline\hline
$\mathbf{0}$ & $\frac{\mathbf{1}}{\mathbf{2}}\mathbf{+}\frac{\mathbf{1}%
}{\mathbf{2}}\operatorname{Re}\left(  \mathbf{U}_{\mathbf{kk}}\right)  $ &
$\frac{\left\vert \mathbf{0}\right\rangle \overset{}{\left(  \frac
{\mathbf{1+U}}{\mathbf{2}}\right)  }\left\vert \mathbf{k}\right\rangle
}{\underset{}{\sqrt{\frac{\mathbf{1}}{\mathbf{2}}\mathbf{+}\frac{\mathbf{1}%
}{\mathbf{2}}\operatorname{Re}\left\langle \mathbf{k}|\mathbf{U}%
|\mathbf{k}\right\rangle }}}$\\\hline
$\mathbf{1}$ & $\frac{\mathbf{1}}{\mathbf{2}}\mathbf{-}\frac{\mathbf{1}%
}{\mathbf{2}}\operatorname{Re}\left(  \mathbf{U}_{\mathbf{kk}}\right)  $ &
$\frac{\left\vert \mathbf{0}\right\rangle \overset{}{\left(  \frac
{\mathbf{1-U}}{\mathbf{2}}\right)  }\left\vert \mathbf{k}\right\rangle
}{\underset{}{\sqrt{\frac{\mathbf{1}}{\mathbf{2}}\mathbf{-}\frac{\mathbf{1}%
}{\mathbf{2}}\operatorname{Re}\left\langle \mathbf{k}|\mathbf{U}%
|\mathbf{k}\right\rangle }}}$\\\hline
\end{tabular}
\]

\item[\fbox{\textbf{Step 5}}] Output the classical bit $\mathbf{b}$ and
\textsc{Stop}
\end{itemize}

\bigskip

Next we formalize the iteration procedure by defining the following quantum subroutine:

\bigskip

\textsc{Approx-Re-Trace}$_{U}\left(  n\right)  $

\qquad\textsc{loop} $k=1..2$

\qquad\qquad\textsc{Approx-Diag-Entry}$(k)=0$

\qquad\qquad\textsc{loop} $\ j=1..n$

\qquad\qquad\qquad$b=$\textsc{QRe}$_{U}\left(  k\right)  $

\qquad\qquad\qquad\textsc{Approx-Diag-Entry}$(k)$ $=$
\textsc{Approx-Diag-Entry}$(k)+(-1)^{b}$

\qquad\qquad\textsc{end loop} $j$

\qquad\textsc{end loop} $k$

\qquad\textsc{Output} $\left(  \text{ \ \textsc{Approx-Diag-Entry}%
}(1)+\text{\textsc{Approx-Diag-Entry}}(2)\text{ }\right)  $

\textsc{End}

\bigskip

As mentioned earlier, quantum subroutines \textsc{QIm}$_{U}\left(  k\right)  $
and \textsc{Approx-Im-Trace}$_{U}\left(  n\right)  $ can be defined in a
similar manor.

\bigskip

We continue by recognizing that there is a certain amount of computational
effort involved in creating the subroutine \textsc{QRe}$_{U}\left(
\ k\ \right)  $. \ \ For this, we need the following formal definition:

\bigskip

\begin{definition}
The \textbf{compilation time} of a quantum algorithm is defined as the amount
of time (computational effort) required to assemble algorithm into hardware.
\end{definition}

\bigskip

Since the compilation time to assemble the gate $U$ is asymptotically the
number of elementary gates $U^{(j)}$ in the product $\prod_{k=1}^{L}\left(
U^{\left(  j\left(  k\right)  \right)  }\right)  ^{\epsilon\left(  k\right)
}$, we have

\bigskip

\begin{theorem}
Let $b$ be a 3-stranded braid, i.e., $b\in B_{3}$, and let $K$ be the knot (or
link) formed from the closure $\overline{b}$ of the braid $b$. \ Then the time
complexity of compiling the braid word $b$ into the quantum subroutine
\textsc{QRe}$_{U}(\quad)$ is $O\left(  L\right)  $, where $L$ is the length of
the braid word $b$, i.e., where $L$ is the number of crossings in the knot (or
link) $K$. \ Moreover, the running time complexity of \textsc{QRe}$_{U}%
(\quad)$ is also $O\left(  L\right)  $. \ The same is true for the quantum
subroutine \textsc{QIm}$_{U}(\quad)$.
\end{theorem}

\bigskip

\begin{corollary}
The quantum subroutine \textsc{Approx-Re-Trace}$_{U}\left(  n\right)  $\ and
\textsc{Approx-Im-Trace}$_{U}\left(  n\right)  $ are each of compile time
complexity $O\left(  nL\right)  $ and of run time complexity $O\left(
nL\right)  $. \ 
\end{corollary}

\bigskip

\begin{theorem}
Let $b$ be a 3-stranded braid, i.e., $b\in B_{3}$, and let $K$ be the knot (or
link) formed from the closure $\overline{b}$ of the braid $b$. \ Let
$\epsilon_{1}$ and $\epsilon_{2}$ be to arbitrary chosen positive real numbers
such that $\epsilon_{2}\leq1$. \ Let $n$ be an integer such that
\[
n\geq\frac{\ln\left(  2/\epsilon_{2}\right)  }{\epsilon_{1}^{2}}\text{ .}%
\]
Then with time complexity $O\left(  nL\right)  $, the quantum algorithm
\textsc{Approx-Re-Trace}$_{U}\left(  n\right)  $ will produce a random real
number $S_{n}^{(\operatorname{Re})}$ such that
\[
\text{Prob}\left(
\begin{array}
[c]{c}%
\ \\
\end{array}
\left\vert S_{n}^{(\operatorname{Re})}-\operatorname{Re}\left(  \mathstrut
trace(U)\right)  \right\vert \geq\epsilon_{1}%
\begin{array}
[c]{c}%
\ \\
\end{array}
\right)  \ \leq\ \epsilon_{2}%
\]
In other words, the probability that \textsc{Approx-Re-Trace}$_{U}\left(
n\right)  $ will output a random real number $S_{n}$ within $\epsilon_{1}$ of
the real part $\operatorname{Re}\left(  trace(U)\right)  $ of the trace
$trace\left(  U\right)  $ is greater than $1-\epsilon_{2}$. \ The same is true
for the quantum subroutine \textsc{Approx-Im-Trace}$_{U}\left(  n\right)  $.
\end{theorem}

\begin{proof}
Let $X_{1},X_{2},\ldots,X_{n}$ be the $n$ random variables corresponding to
the $n$ random output bits resulting from $n$ independent executions of
\textsc{QRe}$_{U}\left(  1\right)  $, and in like manner, let $X_{n+1}%
,X_{n+2},\ldots,X_{2n}$ be the $n$ random variables corresponding to the $n$
random output bits resulting from $n$ independent executions of \textsc{QRe}%
$_{U}\left(  2\right)  $. \ 

\bigskip

Thus, each of the first $n$ random variables have the same probability
$p_{0}^{(1)}$ of being zero and the same probability $p_{1}^{(1)}$ of being
$1$. \ In like manner, the last $n$ of these random variables have the same
probabilities $p_{0}^{(2)}$ and $p_{1}^{(2)}$ of being $0$ or $1$,
respectively. \ Moreover, it is important to emphasize that the $2n$ random
variables $X_{1},X_{2},\ldots,X_{n}$, $X_{n+1},X_{n+2},\ldots,X_{2n}$ are
stochastically independent.

\bigskip

The random variable associated with the random number
\[
\frac{\#0^{\prime}s-\#1^{\prime}s}{n}%
\]
is%
\[
S_{n}^{(\operatorname{Re})}=\sum_{j=1}^{2n}\left(  -1\right)  ^{X_{j}}\text{
.}%
\]
The reader can easily verify the mean $\mu$ of $S_{n}$ is given by
\[
\mu=p_{0}^{(1)}-p_{1}^{(1)}+p_{0}^{(2)}-p_{1}^{(2)}=\operatorname{Re}\left(
U_{11}\right)  +\operatorname{Re}\left(  U_{22}\right)  =\operatorname{Re}%
\left(  trace\left(  U)\right)  \right)
\]

From Hoeffding's inequality\cite[Theorem 2, page 16]{Hoeffding1}, it follows
that%
\[
Prob\left(
\begin{array}
[c]{c}%
\ \\
\end{array}
\left\vert S_{n}^{(\operatorname{Re})}-\operatorname{Re}\left(
trace(U)\right)  \right\vert \geq\epsilon_{1}%
\begin{array}
[c]{c}%
\ \\
\end{array}
\right)  \leq2e^{-2(2N)^{2}\epsilon_{1}^{2}/\left(  \sum_{j=1}^{2n}4\right)
}=2e^{-n\epsilon_{1}^{2}\text{ \ .}}%
\]
Thus, when
\[
n\geq\frac{\ln\left(  2/\epsilon_{2}\right)  }{\epsilon_{1}^{2}}\text{ , }%
\]
we have that%
\[
Prob\left(  \ \left\vert S_{n}^{(\operatorname{Re})}-\operatorname{Re}\left(
trace(U)\right)  \right\vert \geq\epsilon_{1}\ \right)  \leq\epsilon_{2}\text{
\ .}%
\]

In like manner, a similar result can be proved for \textsc{QIm}$_{U}$.
\end{proof}

\bigskip

As a corollary, we have

\bigskip

\begin{corollary}
Let $b$ be a 3-stranded braid, i.e., $b\in B_{3}$, and let $K$ be the knot (or
link) formed from the closure $\overline{b}$ of the braid $b$. \ Let
$\epsilon_{1}$ and $\epsilon_{2}$ be to arbitrary chosen positive real numbers
such that $\epsilon_{2}\leq1$. \ Let $n$ be an integer such that
\[
n\geq\frac{\ln\left(  4/\epsilon_{2}\right)  }{2\epsilon_{1}^{2}}\text{ .}%
\]
Then with time complexity $O\left(  nL\right)  $, the quantum algorithms
\textsc{Approx-Re-Trace}$_{U}\left(  n\right)  $ and \textsc{Approx-Im-Trace}%
$_{U}\left(  n\right)  $ will jointly produce random real numbers
$S_{n}^{(\operatorname{Re})}$ and $S_{n}^{(\operatorname{Im})}$such that
\[
\text{Prob}\left(
\begin{array}
[c]{c}%
\ \\
\end{array}
\left\vert S_{n}^{(\operatorname{Re})}-\operatorname{Re}\left(  \mathstrut
trace(U)\right)  \right\vert \geq\epsilon_{1}\text{ and }\left\vert
S_{n}^{(\operatorname{Im})}-\operatorname{Im}\left(  \mathstrut
trace(U)\right)  \right\vert \geq\epsilon_{1}%
\begin{array}
[c]{c}%
\ \\
\end{array}
\right)  \ \leq\ \epsilon_{2}%
\]
In other words, the probability that both \textsc{Approx-Re-Trace}$_{U}\left(
n\right)  $ and \textsc{Approx-Im-Trace}$_{U}\left(  n\right)  $ will output
respectively random real number $S_{n}^{\left(  \operatorname{Re}\right)  }$
and $S_{n}^{\left(  \operatorname{Im}\right)  }$ within $\epsilon_{1}$ of the
real and imaginary parts of the trace $trace\left(  U\right)  $ is greater
than $1-\epsilon_{2}$. \ 
\end{corollary}

\bigskip

\section{Summary and Conclusion}

\bigskip

Let $K$ be a 3-stranded knot (or link), i.e., a knot formed by the closure
$\overline{b}$ of a 3-stranded braid $b$, i.e., a braid $b\in B_{3}$. \ Let
$L$ be the length of the braid word $b$, i.e., the number of crossings in the
knot (or link) $K$. \ Let $\epsilon_{1}$ and $\epsilon_{2}$ be two positive
real numbers such that $\epsilon_{2}\leq1$. \ 

\bigskip

Then in summary, we have created two algorithms for computing the value of the
Jones polynomial $V_{K}\left(  t\right)  $ at all points $t=e^{i\varphi}$ of
the unit circle in the complex plane such that $\left\vert \varphi\right\vert
\leq\frac{2\pi}{3}$. \ 

\bigskip

The first algorithm, called the \textbf{classical 3-SB algorithm}, is a
classical deterministic algorithm that has time complexity $O\left(  L\right)
$. \ The second, called the \textbf{quantum 3-SB algorithm}, is a quantum
algorithm that computes an estimate of $V_{K}\left(  e^{i\varphi}\right)  $
within a precision of $\epsilon_{1}$ with a probability of success bounded
below by $1-\epsilon_{2}$. \ The \textbf{execution time complexity} of this
algorithm is $O\left(  nL\right)  $, where $n$ is the ceiling function of
$\ \frac{\ln\left(  4/\epsilon_{2}\right)  }{2\epsilon_{1}^{2}}$. \ The
\textbf{compilation time complexity}, i.e., an asymptotic measure of the
amount of time to assemble the hardware that executes the algorithm, is
$O\left(  L\right)  $. \ A pseudo code description of the quantum 3-stranded
braid algorithm is given below.

\bigskip

\textsc{Quantum-3-SB-Algorithm}$\left(  b,\varphi,\epsilon_{1},\epsilon
_{2}\right)  $

\qquad\textsc{Comment:} $b=$ braid word representing a 3-stranded braid s.t.
$K=\overline{b}$

\qquad\textsc{Comment:} $\varphi$ real number s.t. $\left\vert \varphi
\right\vert \leq\frac{2\pi}{3}$

\qquad\textsc{Comment:} $\epsilon_{1}$ lower bound on the precision of the output

\qquad\textsc{Comment:} $\epsilon_{2}$ upper bound on the probability that the

\qquad\textsc{Comment: \ }output is not within precision $\epsilon_{1}$

\qquad\textsc{Comment:} The output of this algorithm is with probability
$\geq1-\epsilon_{2}$

\qquad\textsc{Comment: \ }a complex number within $\epsilon_{1}$ of
$V_{K}\left(  e^{i\varphi}\right)  $

\qquad

\qquad$n=\left\lceil \frac{\ln\left(  4/\epsilon_{2}\right)  }{2\epsilon
_{1}^{2}}\right\rceil $

$\qquad U=$ \ \textsc{Gate-Compile}$\left(  b\right)  $

\qquad\textsc{Approx-Re-Trace}$_{U}=$ \ \textsc{Real-Part-Trace-Compile}%
$\left(  U\right)  $

\qquad\textsc{Approx-Im-Trace}$_{U}=$ \ \textsc{Imaginary-Part-Trace-Compile}%
$\left(  U\right)  $

\qquad\textsc{Approx}ReTr $=$ \textsc{Approx-Re-Trace}$_{U}\left(  n\right)  $

\qquad\textsc{Approx}ImTr $=$ \textsc{Approx-Im-Trace}$_{U}\left(  n\right)  $

\qquad$W=$ \ $Writhe\left(  b\right)  $

\qquad$\theta=-\varphi/4$

\qquad$\delta=-2\ast\cos\left(  2\ast\theta\right)  $

\qquad\textsc{ReExp3} $=\cos\left(  3\ast\theta\ast W\right)  $

\qquad\textsc{ImExp3} $=\sin\left(  3\ast\theta\ast W\right)  $

\qquad\textsc{ReJones} $=$\textsc{ReExp3}$\ast$\textsc{ApproxReTr}%
$-$\textsc{ImExp3} $\ast$ \textsc{ApproxImTr}

\qquad\textsc{ReJones} $=\left(  -1\right)  ^{W}\ast\left(
\begin{array}
[c]{c}%
\ \\
\
\end{array}
\text{\textsc{ReJones}}+\left(  \delta-2\right)  \ast\cos\left(  \varphi\ast
W\right)
\begin{array}
[c]{c}%
\ \\
\
\end{array}
\right)  $

\qquad\textsc{ImJones} $=$\textsc{ImExp3}$\ast$\textsc{ApproxReTr}%
$+$\textsc{ReExp3} $\ast$ \textsc{ApproxImTr}

\qquad\textsc{ImJones}$\ =\left(  -1\right)  ^{W}\ast$ $\left(
\begin{array}
[c]{c}%
\ \\
\
\end{array}
\text{\textsc{ImJones}}-\left(  \delta-2\right)  \ast\sin\left(  \varphi\ast
W\right)
\begin{array}
[c]{c}%
\ \\
\
\end{array}
\right)  $

\qquad\textsc{Output(} \ \textsc{ReJones}, \ \textsc{ImJones} \ \textsc{)}

\textsc{End}

\bigskip

\noindent\textbf{Acknowledgements.} \ This work was partially supported by the
National Science Foundation under NSF Grant DMS-0245588. \ The authors would
like to thank David Cornwell and John Myers for their helpful suggestions and discussions.

\end{document}